\documentclass[preprint,12pt]{article}

\usepackage{fullpage}
\usepackage{amsthm}
\usepackage{amsmath}
\usepackage{amsfonts}
\usepackage{amssymb}
\usepackage{enumerate}
\usepackage{graphics}
\usepackage{graphicx}
\usepackage{color}
\usepackage{multicol}
\usepackage{color}
\usepackage{authblk}

\newtheorem{theorem}{Theorem}
\newtheorem{lemma}{Lemma}
\newtheorem{proposition}{Proposition}

\newtheorem{definition}{Definition}
\newtheorem{corol}{Corollary}
\newtheorem{claim}{Claim}

\newtheorem{remark}{Remark}

\def\MIS{MIS}
\def\spec{\infty}

\def\M{M}

\newcommand{\remove}[1]{}
\newcommand{\Active}{{{\sf Influenced}}}
\newcommand{\A}[1]{\Active[#1]}
\newcommand{\N}{{\mathbb{N}}}
\newcommand{\No}{{\mathbb{N}}_0}

\title{ Spread of Influence in Weighted Networks under Time and Budget Constraints \thanks{An extended abstract of a preliminary version of this paper appeared
in Proceedings of 7th International Conference on Fun with Algorithms (FUN 2014), 
Lecture Notes in Computer Science Vol. 8496, A. Ferro, F. Luccio, P. Widmayer (Eds.), pp. 100-112, 2014.
This work was supported in part by the Slovenian Research Agency (research program P$1$-$0285$ and research projects J$1$-$5433$, J$1$-$6720$, and J$1$-$6743$).
}
}

\author[1]{Ferdinando Cicalese}

\author[2]{Gennaro Cordasco}

\author[1]{Luisa Gargano}

\author[3]{Martin Milani\v{c}}

\author[4]{Joseph Peters} 

\author[1]{U. Vaccaro}

{\small
\affil[1]{\footnotesize
 Department of Computer Science, University of Salerno, Italy, {\texttt{\{cicalese,lg,uv\}@dia.unisa.it}}}
\affil[2]{Department of Psychology, Second University of Naples, Italy, {\texttt {gennaro.cordasco@unina2.it}}}
\affil[3]{University of Primorska, UP IAM and UP FAMNIT, 
SI 6000 Koper, Slovenia,  \texttt{martin.milanic@upr.si}}
\affil[4]{School of Computing Science, Simon Fraser University, Canada, \texttt{peters@cs.sfu.ca}}
}

\begin{document}
\maketitle

\begin{abstract}
Given a 
network represented by a weighted directed graph $G$, 
we consider the problem of finding a bounded cost
set of nodes $S$ such that  the influence spreading from  $S$ in $G$,
within a {given} time bound, 
 is as large as possible. 
The dynamic that governs the spread of influence is the following:
 initially only  elements in $S$ are influenced;
subsequently at each round, the set of influenced  elements is
augmented by all nodes in the network that have a sufficiently large number of
already influenced  neighbors.
We  prove that the problem is NP-hard, even in
simple networks like complete graphs and trees. 
We also derive a series of positive  results.
We  present exact  pseudo-polynomial time algorithms for  general {trees},
that become polynomial time  in case the trees are unweighted. This last result improves on
previously published results. We also design  polynomial time algorithms
for general weighted paths and cycles, and for unweighted complete graphs.
\end{abstract}



\section{Introduction}
\subsection{Motivation} Social influence is the process by which individuals adjust their opinions, 
revise their beliefs, or change their behaviors as a
result of  interactions with other people.
When exposed to  the opinions of peers on a
given issue,  people tend to filter and integrate the 
information they receive and adapt  their own judgements  accordingly (see for instance 
\cite{Yaniv++}). This human tendency to harmonize their own ideas and customs with the 
opinions and behaviors of 
others  \cite{Asch56} may 
occurs for several reasons: 
a) the basic human need to be liked and accepted  by others \cite{Baum+}; 
b) the belief that others, especially a majority group, 
have more accurate and trustworthy information than the individual \cite{S}; 
c) the ``direct-benefit'' effect,  implying that
an individual obtains an  explicit benefit when 
he/she aligns his/her behavior with the behavior of others (e.g., \cite{EK}, Ch.~17).
It has not escaped the attention of advertisers\footnote{and politicians too \cite{Bo+,LKLG,T,RW}} 
that the natural  human tendency to conform
can be exploited in \emph{viral marketing} \cite{LAM}.
Viral marketing 
refers to the spread of information about products and  behaviors,
  and their adoption by people. According to Lately~\cite{D}, ``\emph{the traditional broadcast model of 
advertising-one-way, one-to-many, read-only is increasingly being 
superseded by a vision of marketing that wants, and expects, consumers to spread the word themselves}''.
For what strictly concerns us,  the intent of maximizing the spread of viral information across a network naturally
suggests many interesting  optimization  problems. Some of them  were first articulated in the seminal papers
\cite{KKT-03,KKT-05}, under various adoption paradigms.
The recent monograph \cite{CLC} contains an excellent  
description of the area.
In the next section, we will explain and motivate our model of 
information diffusion, state the problem that we are investigating,
describe our results, and discuss how
they relate to the existing literature.

\newcommand{\nin}[1]{N^{in}(#1)}
\subsection{The Model}
Let $G = (V,E)$ be a directed  graph, $c: V\to \N=\{1, 2, \ldots \}$ be a function assigning  
costs to vertices and   $w: E \to \No=\{0, 1, 2, \ldots \}$ 
be a function assigning weights to edges.
The value $c(v)$ of each vertex $v\in V$ is a measure of how much
it costs to initially convince the member $v$ of the network
to endorse a given product/behaviour.
The weight of an arc $e = (u,v)\in E$, denoted either by $w(e)$ or by $w(u,v)$,  represents the amount of 
influence that  node $u$ exercises on node $v$.  Let $t: V \to \No $ be a
function assigning  thresholds to the vertices of $G$. For each node $v\in V$, the threshold 
value $t(v)$  quantifies  how hard it is to influence  node $v$, in the sense that
 easy-to-influence  elements of the network  have ``low'' $t(\cdot)$ values, and
hard-to-influence  elements have  ``high'' $t(\cdot)$ values \cite{Gr}.

A  process of \emph{influence diffusion} in $G$,  starting at the subset of nodes $S \subseteq V$ (hereafter called 
\emph{target set}), 
is a sequence  of vertex subsets
$$\Active[S,0] \subseteq \Active[S,1] \subseteq \ldots\subseteq \Active[S,\tau] \subseteq \ldots \subseteq V,$$
where
$\Active[S,0] = S$, and such that for all $\tau > 0$,

$$\Active[S,\tau] = \Active[S,\tau{-}1]\cup \Big\{u \,:\, \sum_{v \in \nin{u} \cap 
\Active[S,\tau{-}1] }\!\!\!\!\!\!\!\!\!\!\!\! w(v,u) \geq t(u) \Big\}.$$
Here $\nin{u}=\{v: (v,u)\in E\}$ denotes the set of incoming neighbors of $u$, that is, the set
of nodes in $G$ having a directed arc towards $u$. 
In  words, at each round  $\tau$ a node $u$ becomes influenced 
if  the sum of the influences exercised on $u$ by $u$'s already influenced incoming  neighbors meets or exceeds
 $u$'s threshold $t(u)$.
We  say   that node $u$ is influenced {\em within} round $\tau$ if $u \in  \Active[S,\tau]$;
 $u$ is influenced {\em at} round $\tau>0$ if $u \in  \Active[S,\tau]\setminus \Active[S,\tau-1]$.
 
 \smallskip
The   problem that we introduce and study in this paper  is defined as follows:

\medskip

\noindent
{\sc $(\lambda, \beta)$-Maximally Influencing Set ($(\lambda, \beta)$-MIS)}.\\
{\bf Instance:} A  directed graph $G=(V,E)$, node thresholds $t:V\to \No$, 
vertex costs $c: V\to \N$,
edge influences $w: E \to \No$,  a latency bound $\lambda\in \N$ and a budget $\beta \in \N$.\\
{\bf Objective:} Find a set $S\subseteq V$  such that $c(S)=\sum_{v\in S}c(v)\leq \beta$ and $|\Active[S,\lambda]|$ 
is as large as possible.

\medskip
Notice that the assumption that all vertex costs are positive is without loss of generality. 
{Indeed, if $c(v) = 0$ for some vertex $v$ in the graph, then we can consider a new graph $G'$ obtained from $G$ by eliminating $v$ and by setting} 

{$t'(u) = \begin{cases}{\max\{t(u)-w(v,u),0\}}&{\mbox{if $u$ is an out-neighbor of $v$ in $G$}}\\{t(u)}&{\mbox{otherwise.}}\end{cases}$}

\noindent
{The decrease in the threshold  of the neighbors of $v$ implies that  $\Active[S,\tau]$   in $G'$ 
is equal to $\Active[S\cup \{v\},\tau]$ in $G$, for each $S\subseteq V-\{v\}$ and  $\tau\geq 1$; hence 
$S$ is an optimal solution for $G'$
iff  $S\cup \{v\}$ is an optimal solution for the original instance. }
\remove{Indeed, if $c(v) = 0$ for some vertex $v$ in the graph, then we can transform the problem instance into an equivalent one as follows: set $c(v) = \beta + 1$, set 
$t(v) =  \sum_{u: (u,v)\in E}w(u,v)+1$, and 
for every out-neighbor $u$ of $v$ set $t(u) = \max\{t(u)-w(v,u),0\}$.
In the transformed instance, vertex $v$ will never be chosen to be in an optimal  target set 
due to its high cost, and it will never be influenced due to its high threshold. 
It follows that if $S$ is an optimal solution to the transformed instance, 
then $S\cup \{v\}$ is an optimal solution to the original instance. }
The above transformation can be carried out for all 
vertices of zero cost in time $O(|V|+|E|)$ resulting in an equivalent instance in which all vertex costs are positive.
\medskip

We are  also marginally  interested in the case in which the influence of   each arc and the cost to
initially activate each vertex are  unitary 
(i.e., the  network is unweighted),  and the graph representing the network is symmetric, that is,
$(u,v)\in E$ if and only if $(v,u)\in E$. In this particular scenario, 
studied in the   conference  version  of this paper \cite{CCGMPV14}, the activation process obeys the following simpler rule:
$\Active[S,0] = S$, and  for all $\tau > 0$,
$$\Active[S,\tau] = \Active[S,\tau-1]\cup \Big\{u : \big|N(u)\cap \Active[S,\tau - 1]\big|\ge t(u)\Big\},\,$$
and   the question  is to find a set of vertices $S$ such that $|S|\leq \beta$ and $|\Active[S,\lambda]|$ is as large as possible,
where $\lambda$ is given as input to the problem.

\subsection{Related work}

The above algorithmic problems have roots in the general study
of the \emph{spread of influence} in Social Networks (see  \cite{CLC,EK} and references quoted therein).
For instance, in the area of viral marketing \cite{DNT12,DR-01}, companies  wanting to
promote products or behaviors might  initially  try to target and convince
a few individuals who, by word-of-mouth, can  trigger
a  cascade of influence in the network leading to
an  adoption  of the products by  a much larger number of individuals.

It is clear that the $(\lambda, \beta)$-MIS problem represents
an abstraction of the viral marketing scenario if one makes the reasonable assumption  that an individual
decides to adopt the products if a suitable number of {his/her friends} have adopted
the products. Analogously, the  $(\lambda, \beta)$-MIS problem can describe  other
diffusion problems  arising in sociological, economical,  and biological  networks
(again see   \cite{EK}).
Therefore,  it comes as no surprise that
{special  cases of our} problem (or  variants thereof) have recently
attracted the  attention of the algorithmic community.
We shall limit ourselves here
to discussing the work that is most directly related to ours, and refer the reader to the 
 monographs \cite{CLC,EK} for an excellent overview of the area.
We just mention that our results also seem to be  relevant to other  areas,
like dynamic monopolies \cite{FKRRS-2003,Peleg-02} for instance.

The first authors to study problems of the spread of influence in networks
from an algorithmic point of view were Kempe \emph{et al.} \cite{KKT-03,KKT-05}.
However, they were mostly interested in networks with  randomly chosen thresholds.
Chen \cite{Chen-09} studied the following minimization problem:
given an unweighted  graph $G$ and fixed thresholds $t(v)$, for each vertex $v$ in $G$,
find
a  set of minimum size that eventually influences
all (or a fixed fraction of) the nodes of $G$.
He proved  a  strong inapproximability result that makes unlikely the existence
of an  algorithm with  approximation factor better than  $O(2^{\log^{1-\epsilon }|V|})$.
Chen's result stimulated a series of papers  
\cite{ABW-10,BCNS,BHLM-11,Centeno12,Chiang,Chopin-12,Chun,Chun2,C-OFKR,Ga+,Re,Za} that isolated interesting cases 
in which the problem (and variants thereof) become tractable.

None of the above quoted  papers considered 
the {number of rounds} necessary for the spread of influence in the network, 
the fact that different individuals can exercise different amounts of influence 
on the same person,  or that the cost to initially 
convince  individuals might vary among different members of the network.
However, all of these  questions   correspond to   relevant issues.
Regarding the first question, it is well known that in viral marketing 
it is quite important to spread information quickly.
Indeed, research in Behavioural Economics 
shows that humans make decisions mostly on the basis of very recent events, 
even though they might remember much more \cite{Alba,Chen+}.
Moreover, 
the  conventional idea of long-living viral spread has been challenged by empirical evidence
in several real-life datasets, where it has been found that the  
processes of influence diffusion do not extend after the first  few
initial steps \cite{GWG,WP}.
Therefore, it seems  reasonable to study processes of information diffusion that 
reach the desired goals within a fixed time bound.
Concerning the second point,   it is generally assumed that the   influence that a VIP 
may have  on the behaviour of an individual can be much larger than the amount of influence
exercised on the  same person by a less famous acquaintance, and this phenomenon should be taken into account when
designing effective viral marketing campaigns (e.g., see 
\cite{I+, Ne+}).\footnote{Startups like Klout (\texttt{http://klout.com})  offer a way to quantify 
 the influence of online users of social media. }
 Finally, that different members of the network have different
activation costs (see \cite{Ba+}, for example) is justified by
the reasonable assumption that celebrities or public figures can  charge more for their
endorsements of  products.

The only paper known to us that has studied the spread of influence  with constraints on the number of rounds 
in which the process must be completed (but in unweighted networks and with no costs on vertices)
 is \cite{CCGMV13}. How our results are related to \cite{CCGMV13}
 will be explained   in the next section. 
Paper \cite{Rautenbach2013} studied the problem of finding the smallest
set of vertices that can influence a whole  graph (again, in unweighted networks and with no costs on vertices), where
each vertex has an associated deadline that must be respected by the diffusion 
process.
Finally, we
 point out that Chen's   inapproximability result \cite{Chen-09}
still holds  if the diffusion  process must end in a
bounded number of rounds.

\subsection{The Results}
{In light of Chen's strong inapproximability results \cite{Chen-09}, we feel
 motivated to identify special cases for which our general problems become tractable (i.e., tree, cycle, and clique topologies).
We also feel  that the analyzed networks might approximate some features of real-life networks; for instance, trees emulate hierarchical structure while cliques resemble strongly connected components like communities. Moreover, we believe/hope that our proposed strategies could be useful for the development of novel strategies or heuristics on more elaborate topologies.}

Our first result shows that the {\sc $(\lambda, \beta)$-MIS} problem cannot be 
solved in polynomial time on weighted complete graphs unless $P=NP$. 
On the other hand,  if the graph is complete and 
unweighted, then a  linear time algorithm for the 
{\sc $(\lambda, \beta)$-MIS} problem is quite easy to find. 

In Section \ref{sec-trees} we turn our attention to trees.
We first prove that solving the $(\lambda,\beta)$-MIS problem on  weighted trees
is at least as hard as solving general instances of the well-known NP-hard 
{\sc $0-1$ Knapsack} problem. 
Subsequently, we derive pseudo-polynomial  time algorithms to 
solve the $(\lambda,\beta)$-MIS problem on  weighted trees.
We point out that the paper \cite{CCGMV13} provided 
an algorithmic framework to solve the $(\lambda,\beta)$-MIS problem (and related
ones), in \emph{unweighted} graphs of bounded clique-width. When instantiated on unweighted trees,
the approach of \cite{CCGMV13}  gives algorithms for the $(\lambda,\beta)$-MIS problem
with complexity that is
\emph{exponential} in the parameter $\lambda$, whereas our algorithm, when instantiated on unweighted trees, 
has complexity polynomial in all of the relevant parameters (see Corollary \ref{cor-tree}).

In  Section \ref{path-cycles}, we study the case of  weighted 
paths and cycles and we provide polynomial time algorithms to solve the $(\lambda,\beta)$-MIS problem
on these classes of graphs.

We conclude this discussion  by remarking 
 that, in the very special  case $\lambda=1$, thresholds and costs $t(v)=c(v)=1$ for each vertex $v\in V$,
and edge weights $w(e)=1$ for each $e\in E$,
problems of influence diffusion reduce to well-known domination problems in graphs (and variants thereof).
In particular, when $\lambda=1$,
  $t(v)=c(v)=1$ for each $v\in V$, and $w(e)=1$ for  $e\in E$, 
our $(\lambda,\beta)$-\textsc{MIS} problem reduces to the \textsc{Maximum Coverage}
problem considered in \cite{BGHHJK}.
Therefore, our results can also be seen as far-reaching generalizations of \cite{BGHHJK}.

\section{Complexity of Computing $(\lambda,\beta)$-MIS  in  Complete Graphs }\label{sec:complete}
We prove that the  {\sc $(\lambda, \beta)$-MIS} problem is NP-hard for complete graphs.
{It   was shown  in \cite{Dreyer09} that when $t(v)=d(v)$ for each vertex $v$, where $d(v)$ denotes the in-degree of $v$,
the problem of finding the \emph{minimum} size subset
$S\subseteq V$ such that $\Active[S, \tau]=V$, for some $\tau\geq 0$, is  equivalent to finding 
a minimum size vertex cover of the graph.
Indeed  under the hypothesis that  $t(v)=d(v)$ for each $v\in V$, one has that 
$\Active[S, \tau]=\Active[S, 1]$ for any  $S\subseteq V$ and $\tau>0$; moreover, 
 $\Active[S, 1]=V$ if and only if $S$ is a vertex cover for $G$.}
This observation was used to prove that, for any constant $k\geq 3$,  the above   minimization problem 
cannot be solved in polynomial  time, unless $P=NP$,
in the class of $k$-regular non-bipartite unweighted graphs.
Now, consider the following problem:
\begin{quote}
{\sc $\lambda$-Minimum Size Subset ($\lambda$-MSS)}.\\
{\bf Instance:} A   graph $G=(V,E)$,  thresholds $t:V\to \No$, 
and a 
bound $\lambda\in \N$. 
\\
{\bf Objective:} Find a set $S\subseteq V$ of minimum size such that $\Active[S, \lambda]=V$.
\end{quote} 
Under the assumption that $t(v)=d(v)$ for each $v\in V$,  
a  {minimum} size subset
$S\subseteq V$ such that $\Active[S, \lambda]=V$ (where now $\lambda$ is an input to the problem) would still
correspond to a minimum vertex cover of the graph. Hence, 
The {\sc $\lambda$-MSS} problem cannot be solved in polynomial time unless P=NP.

\begin{theorem}\label{theorem-cg}
The {\sc $(\lambda, \beta)$-MIS} problem cannot be 
solved in polynomial time on weighted complete graphs unless P=NP, even if all vertex costs 
$c(v)$ are equal to 1.
\end{theorem}
\begin{proof}
{
We will prove that if one had a polynomial time 
algorithm to solve the {\sc $(\lambda, \beta)$-MIS} problem  on an arbitrary complete weighted 
graph, then one could also obtain a polynomial time algorithm for the {\sc $\lambda$-MSS} problem.}

{Consider an arbitrary graph $G=(V,E)$ with the thresholds on the nodes given by 
some function $t:V\to \No$. Let $n$ denote the size of $V$. We  construct a complete graph 
$K=(V,F)$ on the same set of vertices $V$, with
weight function on the edges given by
$$ \mbox{for all } (u,v)\in F \quad  w(u,v)= \begin{cases} n +1 \mbox{ \ \ \ if $\{u,v\}  \in E$} \\
                                                1 \mbox{ \ \ \ \ \ \ \ \ \ otherwise,}
																								\end{cases}$$
and for each node $v \in V$, the threshold $t'(v)$ of $v$ in $K$ equal to 
$$t'(v)= (n+1) t(v).$$
One can easily check that any set of initially influenced nodes $S\subseteq V$ generates
the \emph{same} dynamics of influenced nodes in $G$ and $K$, 
that is, for each $\tau\geq 0$ we have that $\Active[S,\tau]$ in $G$ is 
equal  to $\Active[S,\tau]$ in $K$.
{The conclusion of the proof is now clear: if one had a polynomial time algorithm ${\cal A}$  for the 
{\sc $(\lambda, \beta)$-MIS} problem  on arbitrary complete weighted 
graphs, then by using at most $\log |V| $ calls to ${\cal A}$ on the graph $K$,
 one could find in polynomial time a minimum 
size subset $S\subseteq V$ such that $\Active[S, \lambda]=V$ in the graph $G$.} 
This, together with  the hardness of the {\sc $\lambda$-MSS} problem, 
completes  the proof.}
\end{proof}

We now turn our attention to positive results,  restricting our attention to 
complete graphs in which all edge weights are equal. Without loss of generality, we can assume that 
all edge weights are equal to 1.
Since complete graphs are of clique-width at most~$2$, results from \cite{CCGMV13} imply that the
$(\lambda,\beta)$-MIS problem is solvable in polynomial time on such a class of graphs, if $\lambda$ is constant.
Indeed, one can see that  the $(\lambda,\beta)$-\textsc{MIS}  
can be solved in linear time, independently 
of the value of $\lambda$, by using ideas from  \cite{Nichterlein-12}.

If the network is a  complete graph, then  for any subset of vertices $S$ and any round $\tau \geq 1$,
it holds that
$$\Active[S, \tau] = \Active[S, \tau-1] \cup \{v \,:\,t(v) \leq |\Active[S, \tau-1]|\}.$$
Since $\Active[S, \tau - 1] \subseteq \Active[S, \tau],$ we have
\begin{equation} \label{eq:clique-dynamics}
\Active[S, \tau] = S \cup \{v \,:\,t(v) \leq |\Active[S, \tau-1]|\}.
\end{equation}

From (\ref{eq:clique-dynamics}), and by using a standard exchange
argument, one  realizes  that a set $S$
with largest influence is the one containing the nodes with highest
thresholds. Since it is customary in the case of unweighted graphs to make the reasonable
assumption that $t(v)\in \{0, 1, \ldots, n\}$,
the selection of the $\beta$ nodes with highest threshold 
can be done in linear time. Summarizing, we have the following result.

\begin{theorem} \label{theorem:basic}
There exists an  optimal solution $S$
to the $(\lambda, \beta)$-\textsc{MIS}  problem on a complete unweighted graph $G=(V, E)$ that consists
 of the
$\beta$ nodes of $V$ with highest thresholds, and this solution  can be computed in linear time.
\end{theorem}

\section{Complexity of Computing $(\lambda,\beta)$-MIS in  Weighted Trees}\label{sec-trees}

We first show that the $(\lambda,\beta)$-MIS problem on
weighted trees is at least as hard as the well-known 
$0-1$ Knapsack problem, which is defined as follows: 

\medskip
\noindent {\sc $0-1$ Knapsack}.\\
{\bf Instance:} $n$ items, $o_1, o_2, \ldots, o_n$,  where each 
$o_i$ has a profit $p_i$ and weight $w_i,$  a knapsack capacity $W,$
and a  profit bound $P$. \\
{\bf Question:} Does there exist a subset of
items $\{o_{i_1},o_{i_2},\ldots,o_{i_k}\}$,  such that $\sum_{j=1}^{k}
w_{i_j} \leq W$ and $\sum_{j=1}^{k} p_{i_j} \geq P$?



\begin{figure}[!hbp]
\centering
\includegraphics[width=0.7\textwidth]{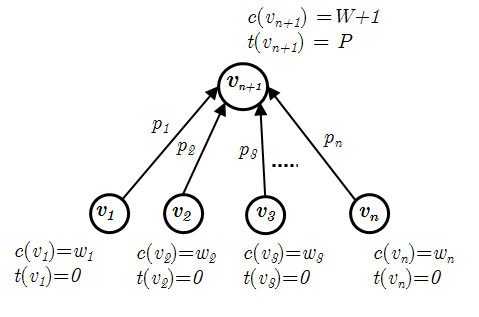}
\caption{The weighted tree $T$.\label{figure:hn}}
\end{figure}

\begin{theorem}\label{theorem-h}
The {\sc $(\lambda, \beta)$-MIS} problem cannot be
solved in polynomial time on weighted star graphs unless P=NP.
\end{theorem}
\begin{proof}
Our reduction will be from the {\sc $0-1$ Knapsack} problem.
Starting from an instance of the {\sc $0-1$ Knapsack} problem, we  build a
weighted tree $T=(V,E)$  as depicted in  Figure
\ref{figure:hn}. The tree $T$ consists of  $n+1$ nodes, one node $v_i$
for each item $o_i$ plus an additional  node $v_{n+1}$.
For each $i=1,2, \ldots, n,$ the node $v_i$ has a directed edge to node
$v_{n+1}$ with weight  $w(v_i,v_{n+1})=p_i$. For each
$i=1,2, \ldots, n,$ the threshold of node $v_i$ is $t(v_i)=0$, 
and the cost of node  $v_i$ is $c(v_i)=w_i$, while 
 $t(v_{n+1})=P$ and $c(v_{n+1})=W+1$. 
It is easy to see that $T$ has a target set $S\subseteq V$ of total
cost at most $W$ such that  $\Active[S,1]=V$ if and only if the instance of the 
{\sc $0-1$ Knapsack} problem has a \textsc{Yes} answer, from which the theorem 
easily follows.


Let $S=\{v_{i_1},v_{i_2},\ldots,v_{i_k}\}\subseteq V$ be a target set
for $T$ such that $\sum_{j=1}^{k} c(v_{i_j}) \leq W$ and
$\Active[S,1]=V.$
Since $c(v_{n+1})=W+1$ we have that $v_{n+1}\notin S$.
The inequality $\sum_{j=1}^{k} c(v_{i_j}) \leq W$ implies that
$\sum_{j=1}^{k} w_{i_j} \leq W$.
The hypothesis that $\Active[S,1]=V$ implies  that $v_{n+1}  \in \Active[S,1]$, that is, 
  $\sum_{j=1}^k w(v_{i_j},v_{n+1}) \geq t(v_{n+1})=P$.  Consequently  $\sum_{j=1}^{k}
p_{i_j} \geq P$.

Conversely, let $K=\{o_{i_1},o_{i_2},\ldots,o_{i_k}\}$ be a subset
of items such that  $\sum_{j=1}^{k} w_{i_j} \leq W$ and
$\sum_{j=1}^{k} p_{i_j} \geq P$. 
Let $S=\{v_{i_1}, \ldots v_{i_k}\}$. We have that $c(S)\leq W$.
Since for each $i=1,2, \ldots, n,$ it holds that $t(v_i)=0$, we also have  
$\{v_1,v_2,\ldots,v_n\} \subseteq \Active[S,1]$.
Moreover, the hypothesis that $\sum_{j=1}^{k} p_{i_j} \geq P$ directly implies
that $\sum_{j=1}^{k} w(v_{i_j}, v_{n+1})=\sum_{j=1}^{k} p_{i_j}\geq P$, 
 consequently the nodes
in $S$ are able to influence the  node $v_{n+1}$ in one step, that is, $\Active[S,1]=V$.
\end{proof}

In the rest of this section we derive a pseudo-polynomial  time 
algorithm for the {\sc $(\lambda, \beta)$-MIS} problem on 
weighted trees.
Let $T = (V,E)$ be a tree having $n$ nodes.
Let us denote by $\Delta$ the maximum indegree of $T$, that is, the quantity
$$\Delta=\max_{v\in V}|\{u: (u,v)\in E\}|$$  
and by   $W$ the quantity	 
$$W=\max_{v\in V}\left\{\sum_{u \in \nin{v}} w(u,v)\right\}.$$
In the following, we will assume that $T$ is rooted at some node $r$. 
For any node $v$ in this rooted tree, we  denote the  subtree rooted at $v$ by $T(v)$,
the set of children of $v$ by $C(v)$, and the parent of $v$ in $T$, for $v \neq r,$ by $p(v)$.
We will develop a dynamic programming algorithm 
that will
 prove the following theorem.

\begin{theorem}\label{theorem-tree}
The {\sc $(\lambda, \beta)$-MIS} problem can be
solved in  time $O(\Delta\, \lambda^2W \beta^3)$ on  a weighted tree with maximum in-degree $\Delta$ and total edge weight $W$.
\end{theorem}

{ The rest of this section is devoted to the design and analysis of the
  algorithm that proves  Theorem \ref{theorem-tree}.
The algorithm   traverses the input  tree $T$  bottom up, in such a way that each node is considered
after all of its children have been processed. 
The basic idea is that the nodes in one subtree of a given node $v$ cannot influence nodes in another subtree without passing through $v$. Moreover, considering a node $v$ and one of its children $u$, there are three possibilities: $v$ influences $u$ (in this case $v$ must be influenced before $u$); $u$ influences $v$ (in this case $u$ must be influenced before $v$); they do not influence each other (the nodes in $T(u)$ cannot influence any other node in $T \setminus T(u)$).
Two particular cases will be considered: 
\begin{itemize}
	\item $v$ belongs to the initial target set $S$. In this case all of the children of $v$ can exploit the influence of $v$ starting in round $1$;
\item $v \notin \A{S, \lambda}$. 
\end{itemize}
In both of these particular  cases, the nodes that belong to different subtrees of $v$ cannot influence each other.
In light of the above observations, for each node $v$, the algorithm solves all possible
{$(\tau, \beta)$-MIS} problems on $T(v)$ for all possible values of $\tau\leq \lambda$ and $\beta\leq b$. Moreover, for some of these values,
we will consider not only the original threshold $t(v)$ of $v$, but also
the decreased value 
\begin{equation}\label{eq:decreasedth}
t'(v)=
\left\{
  \begin{array}{ll}
    \max \{t(v)-w(p(v),v), 0\} & \hbox{if $v\neq r$} \\
    t(v) & \hbox{if $v = r$}
  \end{array}
\right.
\end{equation}
 which we will refer to as the {\em residual threshold}.
The original threshold is used when the nodes in the subtree $T(v)$ are not influenced by $p(v)$ and consequently by any other nodes in $T\setminus T(v)$. The residual threshold is used when $p(v)$ influences $v$. In this case the strategy must guarantee that $p(v)$ will be influenced before $v$.
}

In the following,  we assume without loss of generality that $$0\leq t(u)\leq W(u) + 1,$$
where $W(u)=\sum_{v \in \nin{u}} w(v,u),$ holds for all nodes $u\in V$ (otherwise, we can set
$t(u)=W(u)+ 1$ for every node $u$ with threshold exceeding $W(u) + 1$ without
changing the problem).
\begin{definition} \label{defi:MIS}
For each node $v\in V$,  integers $b =  0, 1, \dots, \beta$,
$t\in \{t'(v),t(v) \}$, and $\tau \in\{0,1,\ldots,\lambda,\spec\}$, let us
denote by $\MIS[v,b,\tau,t]$   the maximum number of nodes that can be influenced in $T(v)$,
in at most $\lambda$ rounds, starting with a target set $S \subseteq V(T(v))$, assuming that

\begin{itemize}
\item the target set is of total cost at most $b$, that is, $c(S) \leq b$;
\item the threshold of $v$ is $t$, and for every $u\in V(T(v))\setminus\{v\}$, the threshold of $u$ is $t(u)$;
\item the parameter $\tau$ is such that
\begin{align}
 & 1) \mbox{ if $\tau=0$ then $v$ must belong to the target set,}\label{eq-case1}\\
\nonumber & 2) \mbox{ if $1\leq\tau\leq \lambda$ then $v$ is not in the target set and the influence of $v$'s } \\
\nonumber & \ \ \ \mbox{children at round  $\tau-1$ is sufficiently large to activate $v$ at round $\tau$,}\\
& \ \ \ \mbox{that is $\sum_{u \in C(v) \cap \A{S,\tau-1}} w(u,v) \geq t$; } \label{eq-case2}\\
 & 3) \mbox{ if $\tau=\infty$  then $v$ is not influenced within round $\lambda$.}\label{eq-case3}
\end{align}
\end{itemize}

We define $\MIS[v,b,\tau,t]= -\infty$ when
the above problem is infeasible.
For instance, if $\tau=0$ and $b<c(v)$  we have
$\MIS[v,b,0,t]= -\infty$.

Denote  by $S(v,b,\tau,t)$ any target set $S \subseteq V(T(v))$ attaining the value $\MIS[v,b,\tau,t]$ (in case of feasible instances).
\end{definition}
We notice that in the above definition,
 if $1\leq \tau\leq \lambda$, then    the assumption that $v$ has threshold $t$ implies that
 $v$ is influenced by round $\tau$
and it is  able to start influencing its neighbors no later than at round $\tau +1$.\footnote{Notice that this does not exclude the case that
$v$ becomes an influenced node at some round before $\tau' < \tau$.}
The value
$\tau=\spec$ means that
$v$ could be either influenced after round $\lambda$ or not influenced at all.
\begin{remark}
It is worthwhile mentioning that  $\MIS[v,b,\tau,t]$ is monotonically non-decreasing in $b$  and non-increasing in $t$.
However, $\MIS[v,b,\tau,t]$ is not necessarily  monotone in $\tau$.
\end{remark}
{Indeed, partition the set $C(v)$ into two sets: $C'(v)$, which contains the $t$ children that influence $v$, and $C''(v)$,  which contains the remaining $|C(v)|-t$ children that may be influenced by $v$. 
A small value of $\tau$ may require a higher budget on subtrees rooted at a node $u \in C'(v)$, and may save some budget on the remaining subtrees; the opposite happens for a large value of $\tau$. 
An example is depicted in Figure \ref{figure:mis}. In the example, all of the node costs $c(\cdot)$ and edge weights $w(\cdot)$ are equal to $1$.  The table reports the value of $\MIS[v,b,\tau,1]$ for each $b\in \{0,1\}$ and $\tau\in\{0,1,2,\spec\}$. }

\begin{figure}[!hbp]
\centering
\includegraphics[width=0.7\textwidth]{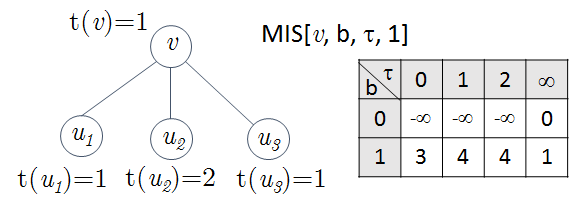}
\caption{A tree $T(v)$ (left) and the value of $\MIS[v,b,\tau,1]$ for each $b\in \{0,1\}$ and $\tau\in\{0,1,2,\spec\}$.\label{figure:mis}}
\end{figure}

The maximum number of nodes in $T$ that can be influenced within round $\lambda$ with
any (initial) target set of cost at most $\beta$
can then be obtained by computing
\begin{equation}\label{eq-mas}
\max_{\tau \in \{0,1,\ldots,\lambda,\spec\}} \ \MIS[r,\beta,\tau,t(r)].
\end{equation}
{We compute this quantity
in Lemma~\ref{lemma:MIS-ABC} by decomposing
$$\max_{\tau \in \{0,1,\ldots,\lambda,\spec\}} \ \MIS[v,b,\tau,t],$$
for 
each $v \in V,$  each $b=0,1,\ldots,\beta$, 
and each $t \in \{t'(v),t(v)\}$, 
into a maximum of three other values which  are successively and separately computed  in Lemmata~\ref{lemma1}--\ref{lemma3}.}

%


We  proceed in a  bottom-up fashion on the tree, so that the computation of the various values $\MIS[v,b,\tau,t]$ for a node $v$ is done after all of
the values for $v$'s children are known.

For each leaf node $\ell$ we have
\begin{equation}\label{eq-casel}
\MIS[\ell,b,\tau,t] =  \begin{cases} 1 & \mbox{ if } (\tau=0 \mbox{ AND } b\geq c(\ell)) \mbox{ OR } (t=0 \mbox{ AND }1\leq \tau \leq \lambda )
\\
0 & \mbox{ if } \tau=\infty
\\
-\infty &  \mbox{otherwise.} \end{cases}
\end{equation}
Indeed, a leaf $\ell$ gets influenced, in the one-node subtree $T(\ell)$, only when either  $\ell$ belongs to the target set ($\tau=0$) and the budget is  sufficiently large ($b\geq c(\ell)$)  or  the threshold is zero (either $t=t(\ell)=0$ or $t=t'(\ell)=0$) independently of the number of rounds.


For any internal node $v$, we show how to compute each  value $\MIS[v,b,\tau,t]$   in time $O(d(v)W(v)\lambda\beta^2)$, where $d(v)$ denotes the in-degree of $v$.


It will be convenient  to analyze the behavior of $MIS[v, b, \tau, t]$ by dividing the possible values of $\tau$ into three cases, according to whether
$\tau = 0$, $\tau \in \{1, \dots, \lambda\}$, or $\tau = \infty.$

To this aim, we will now define three functions, which will be useful for the analysis and the computation of $\MIS[\cdot,\cdot,\cdot,\cdot]$.

In the following we shall also assume that an order has been fixed on the children of any node $v$, that is,  if $v$ has $d$ children we
denote them as $v_1, v_2, \dots, v_d,$ according to the fixed order. Also, we define $F(v, i)$ to be the forest consisting of the subtrees rooted at the
first $i$ children of $v,$ i.e., $F(v, i) = T(v_1) \cup \cdots \cup T(v_i).$ We will also use $F(v,i)$ to denote the set of vertices it includes.

\begin{definition} \label{def:Amax}
Let $v$ be a vertex with $d$ children. For $i=1,\ldots, d$ and $j=0,\ldots, \beta-1$, let $A_v[i,j]$ be the
maximum number of nodes that can be influenced, within $\lambda$ rounds,
in $F(v,i)$ by an influence diffusion process in $T(v)$,
assuming that the target set contains  $v$ and a subset of nodes of $F(v,i)$ of total cost at most $j$.
\end{definition}


\begin{proposition} \label{prop:Amax}
For each vertex $v$ with $d$ children, each $b=0,1,\ldots, \beta$, and each
$t\in \{t(v),t'(v)\}$
, it holds that
\begin{equation}\label{eq-Amax}
\MIS[v,b,0,t] = \begin{cases} 1+A_v[d,b-c(v)] & \mbox { if } b\ge c(v)\\ -\infty & otherwise. \end{cases}
\end{equation}
\end{proposition}
\begin{proof}
By Definition \ref{defi:MIS}, if $b < c(v)$ then the statement is trivially true. Otherwise, the statement directly 
follows from Definitions \ref{defi:MIS} and \ref{def:Amax}. In fact
we have
\begin{eqnarray*}
MIS[v,b,0,t] &=& \max_{\substack{S \subseteq T(v) \\ v \in S, c(S) \leq b}} |\A{S, \lambda} \cap T(v)| \\
&=& 1+ \max_{\substack{v\in S \subseteq T(v) \\ c(S \cap F(v, d)) \leq b-c(v)}} |\A{S, \lambda} \cap F(v, d)| \; \\
&&        \qquad\qquad\qquad\qquad \qquad\quad                         (\mbox{since } v \in S\subseteq \A{S, \lambda}) \\
&=& 1 + A_v[d, b-c(v)]\, \quad \qquad (\mbox{by definition of } A_v[\cdot, \cdot]).
\end{eqnarray*}
\end{proof}

\begin{definition} \label{def:Bmax}
Let $v$ be a vertex with $d$ children and let $\tau = 1, \dots, \lambda.$
For $i=1,\ldots, d$, $j=0,1,\ldots, \beta$, and $k = 0, 1, \dots, t(v),$ we define
$B_{v, \tau}[i,j,k]$ (resp.\  $B_{v, \tau}[\{i\},j,k]$) to be the maximum number of nodes that can be influenced, within $\lambda$ rounds,
by any influence diffusion process in $F(v,i)$ (resp.\ $T(v_i)$)  assuming that
\begin{itemize}
\item the target set $S$ is contained in $F(v,i)$ (resp.\ $T(v_i)$) and is of cost at most  $j$\,,
\item at  time $\tau+1$ the threshold of $v_{\ell}$ becomes   $t'(v_{\ell}),$ for each $\ell = 1, \dots, i$\,, and
\item $\displaystyle{\sum_{\substack{u \in \{v_1, \dots, v_i\} \\ u \in \A{S, \tau-1}}} w(u,v) \geq k}$\,.
\end{itemize}
We also define $B_{v,\tau}[i,j,k]= -\infty$ (resp.\  $B_{v, \tau}[\{i\},j,k]=-\infty$) when the above constraints are not satisfiable.

\medskip
\end{definition}

Hence, $B_{v,\tau}[\{i\},j,k]$ is the same as $B_{v,\tau}[i,j,k]$ but computed on  the subtree $T(v_{i})$ instead of the forest $F(v, i)$. Since $F(v, 1)=T(v_{1})$,
 as a particular case,  we have $B_{v, \tau}[1,j,k] =  B_{v, \tau}[\{1\},j,k].$

\begin{proposition} \label{prop:Bmax}
For each vertex $v$ with $d$ children, each $b=0,1,\ldots, \beta$,  each $\tau = 1, \dots, \lambda,$ and each
$t\in \{t(v),t'(v)\}$, it holds that
\begin{equation}\label{eq-Bmax}
\MIS[v,b,\tau,t] \geq 1+B_{v, \tau}[d, b, t].
\end{equation}
\end{proposition}
\begin{proof}
Let $S$ be a target set achieving $B_{v,\tau}[d,b,t].$
Then  $B_{v,\tau}[d,b,t]$ is the number of influenced nodes within $\lambda$ rounds,
 when the influence diffusion process is run on $F(v, d)$ starting with
$S$. We recall that, by definition, 
  the following conditions are satisfied.
\begin{enumerate}
\item $S \subseteq F(v, d)$ and  $c(S)\leq b$
\item $\displaystyle{\sum_{\substack{i = 1, \dots, d \\ v_i \in \A{S, \tau-1}}} w(v_i, v) \geq t}$
\item from round  $\tau+1$ the threshold of $v_{\ell}$ is decreased to $t'(v_{\ell}),$ for each $\ell = 1, \dots, d$
\end{enumerate}

Now if we use the same target set $S$ in the subtree $T(v)$ with the original thresholds,
 except for $t(v) = t$ we get that $v$ is influenced within time $\tau$ as a consequence of condition 2.
We observe that $\MIS[v, b, \tau, t]$ is the largest possible size achievable for
$\A{S, \lambda}$ under condition 1,  and the condition that $v$ is influenced within round $\tau$.
Finally, considering that the set of influenced vertices contains $v$, we have (\ref{eq-Bmax}).
\end{proof}

\begin{definition} \label{def:Cmax}
Let $v$ be a vertex with $d$ children.
For $i=1,\ldots, d$ and $j=0,\ldots, \beta$, let $C_v[i,j]$ be the maximum number of nodes that can be influenced,
within $\lambda$ rounds, by an influence diffusion process in $F(v, i)$ assuming that the target set $S\subseteq F(v, i)$ is of cost at most $j$.
\end{definition}

\begin{proposition} \label{prop:Cmax}
For each vertex $v$ with $d$ children, each $b= 0,1,\ldots, \beta$, and each
$t\in \{t(v),t'(v)\}$ such that there exists a target set $S \subseteq F(v,d)$ with $c(S) \leq b$ and $v \not \in \A{S, \lambda}$, it holds that
 \begin{equation} \label{eq-cmaxmis}\MIS[v,b,\spec,t] = C_v[d,b].\end{equation}.
\end{proposition}

\begin{proof}
We have
 \begin{eqnarray}
 \MIS[v, b, \infty, t] &=& \max_{\substack{S \subseteq F(v,d) \\ c(S) \leq b \\ v \not \in \A{S, \lambda}}} |\A{S, \lambda} | \label{C:1}\\
 &=& \max_{\substack{S \subseteq F(v,d) \\ c(S)\leq b \\ v \not \in \A{S, \lambda}}} \sum_{i=1}^d |\A{S, \lambda} \cap T(v_i)| \label{C:2}\\
 &=& \max_{\substack{S \subseteq F(v,d) \\ c(S) \leq b}} \sum_{i=1}^d |\A{S \cap T(v_i), \lambda} \cap T(v_i)| \label{C:3} \\
 &=& C_v[d,b], \label{C:4}
\end{eqnarray}
where
(\ref{C:2}) follows from (\ref{C:1}) because, assuming $v$ is not influenced, there is no influence spreading between $T(v_i)$ and $T(v_j)$ for any $1 \leq i, j\leq d$ with $i\neq j$;
(\ref{C:3}) follows from (\ref{C:2}) because if there is no influence spreading between two different  subtrees of $F(v,d)$, then the set of influenced nodes can be computed independently in each subtree; finally (\ref{C:4}) follows from (\ref{C:3}) by the definition of $C_v[d,b].$
%
\end{proof}

\begin{lemma} \label{lemma:MIS-ABC}
For each vertex $v$ with $d$ children, each $b\in \{0,1,\ldots, \beta\}$, and each
$t\in \{t(v),t'(v)\}$, it holds that
\begin{equation}\label{eq:new}
\max_{\tau\in\{0,1,\ldots, \lambda,\infty\}} \MIS[v, b, \tau, t] =\max\left \{1+ A_{v}[d, b-c(v)],  1+  \max_{1\leq \tau_1 \leq \lambda} B_{v, \tau_1}[d,b,t], C_v[d,b]\right\}.
\end{equation}
Moreover, the knowledge  of quantities   $A_{v}[d, b-c(v)]$, $B_{v, \tau}[d,b,t]$,  and $C_v[d,b]$ 
also allows the computation of $\MIS[v, b, \tau, t]$ for each value of $\tau \in \{0,1\ldots,\lambda,\infty\}$. 
\end{lemma}
\begin{proof}
{For notational convenience, let $\M$\ denote the right hand side of (\ref{eq:new}).}
First, suppose that there exists a target set $S \subseteq F(v,d)$ with $c(S) \leq b$ such that $v \not \in \A{S, \lambda}$.
Then, by Propositions \ref{prop:Amax}, \ref{prop:Bmax}, and \ref{prop:Cmax}, we have
$$\max_{\tau\in\{0,1,\ldots, \lambda,\infty\}} \MIS[v, b, \tau, t] \geq M.$$
Now, suppose that for every target set $S \subseteq F(v,d)$
with $c(S) \leq b$ we have $v \in \A{S, \lambda}$. 
We claim that in this case we have
$$C_v[d,b]\le B_{v,1}[d,b,t]+1\,.$$
Indeed, let $S$ be a target set achieving $C_v[d,b]$.
Running the influence diffusion process on $F(v,d)$ with $S$ is equivalent to
running the process on $T(v)$ and ignoring the influence of $v$ on its children (which can be modelled by setting
$w(v,v_i) = 0$ for each $i= 1, \dots, d$).
It can be seen that, given the target set $S$,
increasing the weights on some edges cannot decrease the
number of nodes in $F(v,d)$ influenced within $\lambda$ rounds.
This implies that $C_v[d,b]$ is not greater than the number of nodes
in $F(v,d)$ influenced within $\lambda$ rounds when the influence diffusion process is run from
$S$ in the tree $T(v)$ with the original threshold, which, in turn, does not exceed
$B_{v,1}[d,b,t]+1$.

Summarizing the above two cases, we see that in any case we have
\begin{equation} \label{eq:global_geq}
\max_{\tau\in\{0,1,\ldots, \lambda,\infty\}} \MIS[v, b, \tau, t] \geq \M.
\end{equation}

To see that the converse inequality
\begin{equation} \label{eq:global_leq}
\max_{\tau\in\{0,1,\ldots, \lambda,\infty\}} \MIS[v, b, \tau, t] \leq \M
\end{equation}
also holds, let $\tau^* \in {\rm argmax}_{\tau\in\{0,1,\ldots, \lambda,\infty\}} \MIS[v, b, \tau, t].$

{If $\tau^* =0$, we have $1+ A_{v}[d, b-c(v)] =  \MIS[v, b, \tau^*, t]$ by Proposition \ref{prop:Amax}.
Analogously,  if $\tau^* = \infty$ then $v \not \in \A{S,\lambda}$ for the
target set $S$ achieving $\MIS[v, b, \tau^*, t]$ by Proposition \ref{prop:Cmax}. Then, by Definition \ref{defi:MIS}, we have $C_{v}[d, b] =  \MIS[v, b, \tau^*, t].$
Hence, in both of the above cases, the desired inequality (\ref{eq:global_leq}) also holds {\em a fortiori}.}

Let us now assume that $\tau^* \in \{1, \dots, \lambda\}.$ Let $S \subseteq F(v,d)$ be a
target set of cost at most $b$ which achieves $\MIS[v, b, \tau^*, t].$
Let $\tilde{\tau} \leq \tau^*$ be the minimum positive integer such that $v \in \A{S, \tilde{\tau}}.$
Therefore, no influence is spread from $v$ towards the  subtrees of $F(v, d)$ before round $\tilde\tau.$
Let $S_i = S \cap T(v_i).$ The previous observation implies that 
$\A{S, \tau} \cap T(v_i) = \A{S \cap T(v_i), \tau}$ for each $\tau \leq \tilde{\tau},$ i.e.,
the spread of influence within $T(v_i)$
until round $\tau$ is only determined by the set $S_i.$ From $\tilde\tau$ on, in $T(v_i)$
the fact that $v$ is influenced is equivalent to saying that the threshold of $v_i$ has been decreased to
$t'(v_i).$

Formally, this means that
$$\left|\bigcup_i \left(\A{S, \lambda} \cap T(v_i)\right)\right| \le B_{v,\tilde{\tau}}[d,b,t]$$
hence  we have
\begin{eqnarray} 
\nonumber \max_{\tau \in \{1,\ldots,\lambda\}} \MIS[v, b, \tau, t] &=& \MIS[v, b, \tau^*, t] \\
\nonumber &=& 1 + \left|\bigcup_i \left(\A{S, \lambda} \cap T(v_i)\right)\right| \\
\nonumber &\le& 1 + B_{v,\tilde{\tau}}[d,b,t]\\
          &\leq& 1 + \max_{\tau \in \{1,\ldots,\lambda\}} B_{v,\tau}[d,b,t].
\label{eq:maxtauMIS}
\end{eqnarray}

This concludes the proof of (\ref{eq:global_leq}) that, together with \eqref{eq:global_geq},  yields the desired result, i.e., 
formula (\ref{eq:new}). 

Notice that the above reasoning  proves a slightly more general fact, that is, the inequality
\begin{equation}\label{eq:new3}
 \max_{\tau' \in \{1,\ldots,\tau\}} \MIS[v, b, \tau', t]\leq 1 + \max_{\tau' \in \{1,\ldots,\tau\}} B_{v,\tau'}[d,b,t]
\end{equation}
for any $\tau=1,2,\ldots,\lambda$. 
Formula (\ref{eq:new3}),  together with Proposition \ref{prop:Bmax}, allows us to conclude that
\begin{equation}\label{eq:new4}
 \max_{\tau' \in \{1,\ldots,\tau\}} \MIS[v, b, \tau', t]= 1 + \max_{\tau' \in \{1,\ldots,\tau\}} B_{v,\tau'}[d,b,t]
\end{equation}
for any $\tau=1,2,\ldots,\lambda.$



Moreover, for each $\tau=1,2,\ldots,\lambda -1$ we also have $\MIS[v, b, \tau, t] \le \MIS[v, b, \tau+1, t]$. Therefore by comparing  $\max_{\tau'\in\{1,...,\tau\}} \MIS[v,b,\tau',t]$ and $\max_{\tau'\in\{1,...,\tau+1\}} \MIS[v,b,\tau',t]$, we are also able to compute
 $\MIS[v, b, \tau, t]$ for each value of $\tau=1,2,\ldots,\lambda.$ Recalling that, for  $\tau=0$ and $\tau=\spec$, the value of $\MIS[v, b, \tau, t]$ is easily determined using Propositions \ref{prop:Amax} and  \ref{prop:Cmax}, respectively, we have that the knowledge  of quantities   $A_{v}[d, b-c(v)]$, $C_v[d,b]$, and $B_{v, \tau_1}[d,b,t]$  for each $\tau_1=1,\ldots,\lambda$ 
also allows the computation of $\MIS[v, b, \tau, t]$ for each value of $\tau \in \{0,1\ldots,\lambda,\infty\}$. 
\end{proof}

\bigskip
%

\begin{lemma}\label{lemma1}
For each  vertex $v$,
for each $b=0,1,\ldots, \beta$, and for each $t\in \{t(v),t'(v)\}$,
the quantity $\MIS[v,b,0,t]$ can be computed in time $O(d \lambda b^2),$ where $d$ is the number of children of $v$.
\end{lemma}
\begin{proof}
%

If $b<c(v)$ then the problem is infeasible and $\MIS[v,b,0,t]=-\infty$. Otherwise, by Proposition \ref{prop:Amax}, it is enough to show that we can compute $A_v[d,b-c(v)]$ in the claimed bound.
This will be a consequence of the following recursive characterization of $A_v[i,j]$, for each $i=1,2,\ldots,d$ and  $j=0,1,\ldots,b-c(v)$.

\medskip

\noindent
For  $i=1$, we have
\begin{equation} \label{eq:A-1}
A_v[1,j]=
\max_{\tau_1,t_1} \{\MIS[v_1,j,\tau_1,t_1]\},
\end{equation}
where
$$1.~\tau_1 \in \{0,\ldots,\lambda,\spec\} \qquad 2.~t_1\in \{ t(v_1), t'(v_1)\} \qquad 3.~{\rm{ if }}~t_1=t'(v_1) {\rm{~then}}~\tau_1 \geq  1.$$

To see that the left hand side of (\ref{eq:A-1}) is at least as large as the right hand side we observe that
\begin{eqnarray*}
A_v[1,j] &=& \max_{\tau_1 \in \{0, 1, \dots, \lambda, \infty\}} MIS[v_1, j, \tau_1, t'(v_1)] \\
&\geq&
\max \{MIS[v_1, j, 0, t(v_1)], \max_{\tau_1 \in \{1, \dots, \lambda, \infty\}} MIS[v_1, j, \tau_1, t'(v_1)]\} \\
\end{eqnarray*}
and the last expression is exactly the right hand side of (\ref{eq:A-1}).

For the inequality in the other direction, let $S \subseteq T(v)$ be a target set (of cost at most $j$) achieving $A_v[1, j].$
If $v_1 \in S$ then the node $v$ does not have any effect on the nodes influenced in $T(v_1).$ Hence we have
\begin{equation} \label{eq:A-2}
A_v[1, j] = | \A{S, \lambda} \cap T(v_1) | \leq MIS[v_1, j, 0, t(v_1)|.
\end{equation}
{If $v_1 \not \in S$ then let $\tau^* \in \{1, \dots, \lambda, \infty\}$ be the minimum positive integer such that  $v_1$ is influenced at time $\tau^*$ because of $S$;
$v_1 \not \in \A{S, \lambda}$ then $\tau^* = \infty.$}
Then, since in the definition of $A_v[1, j]$ we assume that $v$ is influenced, or equivalently that the threshold of $v_1$ is reduced to $t'(v_1)$, we have that
 \begin{equation} \label{eq:A-3}
 A_v[1, j] =   | \A{S, \lambda} \cap T(v_1) | \leq \MIS[v_1, j, \tau^*, t'(v_1)],
 \end{equation}
 where the last inequality follows by observing that, in the middle expression, the role of $v$ is only to reduce the threshold of $v_1$ to $t'(v_1).$

The last expressions in both (\ref{eq:A-2})-(\ref{eq:A-3}) contribute to the $\max$ on the right hand side of (\ref{eq:A-1}),
 hence this is also an upper bound for
 $A_v[1,j]$.

\bigskip

\noindent
For $i>1$, we will show that
\begin{equation} \label{eq:A-4}
A_v[i,j]=\max_{0\leq a \leq j} \left \{ A_v[i-1,a] + \max_{\tau_i,t_i} \{\MIS[v_i,j-a,\tau_i,t_i]\} \right \}
\end{equation}
where
$$1.~\tau_1 \in \{0,\ldots,\lambda,\spec\} \qquad 2.~t_1\in \{ t(v_1), t'(v_1)\} \qquad 3.~{\rm{ if }}~t_i=t'(v_i) {\rm{~then}}~\tau_i \geq  1.$$

%

This means that we can compute the quantity $A_v[i,j]$ by considering all possible ways of partitioning the budget  $j$ into two values $a$ and $j-a$, recursively solving a subproblem on $F(v,i-1)$ with budget $a$ and a subproblem on $T(v_i)$ with budget $j-a$, and then combining the solutions.

In order to prove (\ref{eq:A-4}) we have

\begin{eqnarray}
A_v[i,j] &=& \max_{\substack{S \subseteq F(v,i) \\ c(S) \leq j}} |\A{S \cup \{v\}, \lambda} \cap F(v,i)|   \label{eqnarr:A-1}\\
&=& \max_{\substack{S \subseteq F(v,i)\\ c(S) \leq j}}
\left\{
	|\A{\left(S \cap F(v,i-1)\right) \cup \{v\}, \lambda} \cap F(v,i-1)| \right.  \nonumber\\
& & \left. ~~~~~~~~~~~~~~~~~+   	|\A{\left(S \cap T(v_i)\right) \cup \{v\}, \lambda} \cap T(v_i)|  \right\}  \label{eqnarr:A-2}\\
&=& \max_{0 \leq a \leq j} \left\{
  \max_{\substack{c(S_1) \le a, S_1 \subseteq F(v, i-1)}}
	|\A{S_1 \cup \{v\}, \lambda} \cap F(v,i-1)| \right. \nonumber  \\
& & \left. ~~~~~~~~~~+  \max_{\substack{c(S_2) \le j-a, S_2 \subseteq T(v_i)}} |\A{S_2 \cup \{v\}, \lambda} \cap T(v_i)
\right\} \label{eqnarr:A-3} \\
&=& \max_{0 \leq a \leq j} \left\{ A_v[i-1, a] + \max_{\tau_i, t_i} MIS[v_i, j-a, \tau_i, t_i] \right\} \label{eqnarr:A-4}
\end{eqnarray}

where
\begin{itemize}
\item (\ref{eqnarr:A-2}) follows from (\ref{eqnarr:A-1}) because the spread of influence between $F(v, i-1)$ to $T(v_i)$  can only happen
via $v,$
\item (\ref{eqnarr:A-3}) is obtained from  (\ref{eqnarr:A-2}) by standard algebraic manipulation, taking into account that $F(v, i-1) \cap T(v_i) = \emptyset,$
\item (\ref{eqnarr:A-4}) follows from (\ref{eqnarr:A-3}) because
$$\max_{\substack{c(S_1) \le a, S_1 \subseteq F(v, i-1)}}
	|\A{S_1 \cup \{v\}, \lambda} \cap F(v,i-1)| = A_v[i-1, a] $$
	holds by definition and, by proceeding in perfect analogy with the proof of (\ref{eq:A-1}), one can prove that
$$\max_{\substack{c(S_2) \le  j-a, S_2 \subseteq T(v_i)}} |\A{S_2 \cup \{v\}, \lambda} \cap T(v_i) |=
   \max_{\tau_i, t_i} MIS[v_i, j-a, \tau_i, t_i]$$ under the conditions
   $$1.~\tau_i\in \{0,\ldots,\lambda,\spec\} \qquad 2.~t_i\in \{ t(v_i), t'(v_i)\} \qquad 3.~{\rm{ if }}~t_i=t'(v_i) \rm{~then}~\tau_i \geq  1.$$
\end{itemize}

From the above recursive formulas, it immediately follows that the computation of  $A_v[d, b-1]$ comprises  $O(d  b)$ values
each of which can be computed recursively in time $O(\lambda b)$.  This together with  (\ref{eq-Amax}) implies that
$\MIS[v,b,0,t]$ can be computed  in time $O(d \lambda b^2)$.
\end{proof}

\medskip

\noindent We now consider the computation of $B_{v, \tau}[d,b,t].$ We prepare two technical lemmata. For this we will rely on
the definition of $B_{v, \tau}[\{i\}, j, k]$ as the restriction of $B_{v, \tau}[i, j, k]$ where the forest $F(v,i)$ is replaced by the 
single subtree $T(v_i).$ 

\begin{lemma} \label{lemma:claim1}
For each vertex $v$ with $d$ children, each $\tau=1,\ldots, \lambda,$ each $i=1, \dots, d,$ and each $j = 0, \ldots, \beta,$  we have
 \begin{equation} \label{eq:B-i1k0}
B_{v, \tau}[\{i\}, j, 0] = \max \left\{
\max_{\tau_i \in\{0,1,\dots, \lambda, \infty\}} \MIS[v_i, j, \tau_i, t(v_i)],
\max_{\substack{\tau_i \in\{\tau+1, \dots, \lambda\} \\ \MIS[v_i, j, \tau_i, t'(v_i)] > \MIS[v_i, j, \tau, t'(v_i)]}}
 \MIS[v_i, j, \tau_i, t'(v_i)] \right\}.
 \end{equation}
\end{lemma}
\begin{proof}
For notational convenience, let $R$ denote the right hand side of (\ref{eq:B-i1k0}).

By definition, if  a target set  $S \subseteq T(v_i)$ achieves the value $B_{v, \tau}[\{i\}, j, 0],$ then  $|\A{S, \lambda} \cap T(v_i)| = B_{v, \tau}[\{i\}, j, 0].$
\begin{itemize}
\item If there is a target set $S \subseteq T(v_i)$ that achieves $B_{v, \tau}[\{i\}, j, 0]$ and  $v_i \not \in \A{S, \lambda} \cap T(v_i)$ then
$|\A{S, \lambda} \cap T(v_i) | \leq \MIS[v_i, j, \infty, t(v_i)] \leq R$\,.
\item  If there is a target set $S \subseteq T(v_i)$ that achieves $B_{v, \tau}[\{i\}, j, 0]$ and $v_i \in \A{S, \tau}$ then
$|\A{S, \lambda} \cap T(v_i) | \leq \MIS[v_i, j, \tau, t(v_i)] \leq R$\,.
\item  If for every target set $S \subseteq T(v_i)$ that achieves $B_{v, \tau}[\{i\}, j, 0]$ it holds that: (i)  $v_i \not \in \A{S, \tau}$, and
(ii) $v_i \in \A{S, \tau'}$ for some $\tau+1 \leq \tau' \leq \lambda$, then we have that for any such $S$ it holds that {$|\A{S, \lambda} \cap T(v_i) | \leq \MIS[v_i, j, \tau', t'(v_i)].$} Moreover, by (i) and (ii) we also have that $\MIS[v_i, j, \tau', t'(v_i)] > \MIS[v_i, j, \tau, t'(v_i)].$ Hence,
$|\A{S, \lambda} \cap T(v_i)| \leq \MIS[v_i, j, \tau', t'(v_i)] \leq R.$
\end{itemize}
The above three cases show that
$B_{v, \tau}[\{i\}, j, 0] \leq R$\,.

\smallskip

To show the inequality in the other direction, we consider two cases according to which of the two $\max$ 
expressions in the right hand side of  (\ref{eq:B-i1k0}) gives $R$.
\begin{itemize}
\item Let $\tau_i \in \{0,1,\dots, \lambda, \infty\}$ be such that $\MIS[v_i, j, \tau_i, t(v_i)] = R.$ Let $S$ be a target set achieving $\MIS[v_i, j, \tau_i, t(v_i)].$
Then the influence diffusion process restricted to $T(v_i)$ and starting with $S$, in $\lambda$ rounds, in each of which the threshold of 
$v_i$ is $t(v_i)$, will influence some set $I$ of size $\MIS[v_i, j, \tau_i, t(v_i)].$ 
Clearly, starting the process with the same set $S$ and reducing the threshold of $v_i$ to $t'(v_i)$ from
round $\tau+1$ can only result in a set of influenced nodes which is a superset of $I.$ Hence, 
$R = \MIS[v_i, j, \tau_i, t(v_i)] \leq B_{v, \tau}[\{i\}, j, 0]$ for each $\tau_i =
0,1,\dots, \lambda, \infty.$
\item  Suppose that $R$ is achieved only by the second component of the $\max$ on the right hand side of (\ref{eq:B-i1k0}), i.e.,
$$\max_{\tau_i \in \{0,1, \dots, \lambda, \infty\}} \MIS[v_i, j, \tau_i, t(v_1)] < R = \MIS[v_i, j, \hat{\tau}, t'(v_i)]$$ for some
$\hat{\tau} \in \{\tau+1, \dots, \lambda\}$ such that  $\MIS[v_i, j, \hat{\tau}, t'(v_i)] > \MIS[v_i, j,  \tau, t'(v_i)].$
Because of the inequality $\MIS[v_i, j, \hat{\tau}, t'(v_i)] > \MIS[v_i, j,  \tau, t'(v_i)]$, there must exist a target set $S$ achieving
$\MIS[v_i, j, \hat{\tau}, t'(v_i)]$ such that, in the influence diffusion process in $T(v_i)$ started with $S$, the vertex $v_i$ is influenced
later than round $\tau.$ Therefore, this influence diffusion process exploits the reduction of the threshold of $v_i$ only after round $\tau,$ which implies that
\begin{eqnarray*}
R = \MIS[v_i, j, \hat{\tau}, t'(v_i)] &=& |\A{S, \lambda} \cap T(v_i) | \\
&\leq&
\max_{\substack{S' \subseteq T(v_i),~c(S') \leq j \\ t(v_i) = t'(v_i) \mbox{\footnotesize{ from round }}\tau+1}} |\A{S', \lambda} \cap T(v_i) | \\
&=& B_{v, \tau}[\{i\}, j, 0].
\end{eqnarray*}
\end{itemize}
In both cases we have $R \leq B_{v, \tau}[\{i\},j,0].$ This together with the previously shown inequality in the other direction completes the proof
of (\ref{eq:B-i1k0}). 
\end{proof}

\begin{lemma} \label{lemma:claim2}
For each vertex $v$ with $d$ children, each $\tau=1,\ldots, \lambda, $ each  $i=1, \dots, d,$ each $ j = 0, \ldots, \beta,$ and each $0 <k \leq w(v_i, v)$,  we have
\begin{equation} \label{eq:B-i1k}
B_{v, \tau}[\{i\}, j, k] =
\max_{\tau_i \in\{0,1, \dots, \tau-1\}}  \MIS[v_i, j, \tau_i, t(v_i)].
\end{equation}
\end{lemma}
\begin{proof}
Let set  $S \subseteq T(v_i)$ achieve the value $B_{v, \tau}[\{i\}, j, k],$ that is,  $|\A{S, \lambda} \cap T(v_i)| = B_{v, \tau}[\{i\}, j, k].$
Since $k > 0$, it means that by time $\tau$ the only child of $v$, namely $v_i,$\footnote{Recall that when we use $B_{v, \tau}[\{i\}, j, k]$,
we refer to the modified tree in which $F(v, d)$ has been replaced by $T(v_i)$. Hence $v$ now has only one child which, abusing notation, we continue to refer to as $v_i$ for the sake of keeping the correspondence with the original tree.}
exerts some influence on $v,$ hence $v_i$ has already
been influenced by time $\tau-1$.
Let $\tau'\in \{0,1,\ldots, \tau-1\}$ denote the minimum round at which $v_i$ gets influenced,
with  $t(v_i)$ being the threshold of $v_i$ at time $\tau'.$
Then $$B_{v, \tau}[\{i\},j,k]  \leq
\max_{\substack{S' \subseteq T(v_i) \\ c(S') \leq j \\ v_i \in \A{S', \tau'}}} |\A{S', \lambda} \cap T(v_i)| = \MIS[v_i, j, \tau', t(v_i)].$$

\medskip

\noindent
For the opposite inequality, let $\tau_i \in \{0,1,\ldots, \tau-1\}$ be such that $\MIS[v_i, j, \tau_i, k]$ achieves the maximum on the right hand side of
(\ref{eq:B-i1k}). Let $S$ be a target set achieving the maximum of $\MIS[v_i, j, \tau_i, k]$. Hence, $v_i \in \A{S, \tau_i} \subseteq
\A{S, \tau-1},$ since $\tau_i \leq \tau-1.$ Therefore, at time $\tau-1$ the influence from $v_i$ to $v$ is $w(v_i, v) \geq k.$
Notice that, since there is only one child of $v,$ namely $v_i$, the condition $\sum_{u \in C(v) \cap \A{S, \tau-1}} w(u,v) \geq k > 0$
is equivalent to requiring $v_i \in \A{S, \tau-1}.$
This implies
$$\MIS[v_i, j, \tau_i, k] \leq
\max_{\substack{S' \subseteq T(v_i) \\ c(S') \leq j \\ v_i \in \A{S', \tau-1}}} | \A{S', \lambda} \cap T(v_i)| = B_{v, \tau}[\{i\},j,k]$$
which provides the desired inequality and completes the proof of (\ref{eq:B-i1k}).
\end{proof}

\begin{lemma}\label{lemma2}
For each vertex $v$, each $b=0,1,\ldots, \beta$, each $t\in\{t(v),t'(v)\}$, and each $\tau=1,\ldots, \lambda$, it is possible to
 compute $B_{v,\tau}[d,b,t]$
 recursively in time $O(d\lambda b^2 t),$ where $d$ is the number of children of $v$.
\end{lemma}
\begin{proof}
We can  compute $B_{v,\tau}[d,b,t]$ by recursively computing the values of $B_{v,\tau}[i,j,k]$ for each $i=1,2,\ldots,d,$ each $j=0,1,\ldots,b,$ and each $k=0,1,\ldots,t,$ as follows.

\medskip

\noindent
Let $i=1$. We split this case into three subcases according to the value of $k.$

For $k = 0$ we have $B_{v, \tau}[1, j, 0] = B_{v, \tau}[\{1\}, j, 0],$ hence by Lemma \ref{lemma:claim1}, we have
\begin{equation} \label{eq:B-1-0}
B_{v, \tau}[1, j, 0] = \max \left\{
\max_{\tau_1 \in\{0,1,\dots, \lambda, \infty\}} \MIS[v_1, j, \tau_1, t(v_1)],
\max_{\substack{\tau_1 \in\{\tau+1, \dots, \lambda\} \\ \MIS[v_1, j, \tau_1, t'(v_1)] > \MIS[v_1, j, \tau, t'(v_1)]}}
 \MIS[v_1, j, \tau_1, t'(v_1)] \right\}.
 \end{equation}

\bigskip

For $0 < k \leq w(v_1,v)$ we have $B_{v, \tau}[1, j, k] = B_{v, \tau}[\{1\}, j, k],$ hence by Lemma \ref{lemma:claim2}, we have
\begin{equation} \label{eq:B-1-k}
B_{v, \tau}[1, j, k] =
\max_{\tau_1 \in\{0,1, \dots, \tau-1\}}  \MIS[v_1, j, \tau_1, t(v_1)].
\end{equation}

\bigskip

Finally, if $k>w(v_1,v)$, then clearly $B_{v, \tau}[1, j, k] =-\infty$.

\bigskip

\noindent
Let $i \in \{2,\dots, d\}$. In order to compute $B_{v, \tau}[i, j, k],$ proceeding as in Lemma \ref{lemma1},
we consider all possible ways of partitioning the budget  $j$ into two values $a$ and $j-a$.
The budget $a$ is used in $F(v,i-1)$, while the remaining budget $j-a$ is assigned to $T(v_i)$.
Moreover, in order to ensure that
\begin{equation} \label{eq:cond-k}
\sum_{\substack{\ell \in\{1, \dots, i\} \\ v_i \in  \A{S,\tau - 1}}} w(v_{\ell},v) \geq k,
\end{equation}
 there are two possibilities to consider:

\begin{description}
	\item[{\textsc{I}})] $\displaystyle{\sum_{\substack{\ell \in\{1, \dots, i-1\} \\ v_i \in  \A{S,\tau - 1}}} w(v_{\ell},v) \geq k}$, i.e., the condition on the influence
	brought to $v$ from $v_1, \dots, v_i$ at time $\tau-1$ is already satisfied by $v_1, \dots, v_{i-1}.$
	In this case we have no constraint on whether and when $v_i$ is influenced, and we can use a reduced threshold from round $\tau+1$;
	\item[{\textsc{II}})] Otherwise, $v_i$ has to contribute to condition (\ref{eq:cond-k}). Hence, $v_i$ has to be influenced before round $\tau$ and  cannot use the reduced threshold.
\end{description}

Therefore, for $i> 1$ and for each $0 \leq j \leq \beta$ and $0 \leq k \leq t,$ we can compute
$B_{v,\tau}[i,j,k]$ using the following formula:
\begin{equation} \label{eqBmax-tot}
B_{v,\tau}[i,j,k]=\max \Big\{ B_{v,\tau}^{\textsc{i}}[i,j,k], B_{v,\tau}^{\textsc{ii}}[i,j,k] \Big\}\,,
\end{equation}
where $B_{v,\tau}^{\textsc{i}}[i,j,k]$ and $B_{v,\tau}^{\textsc{ii}}[i,j,k]$ denote the corresponding optimal values of
the two restricted subproblems.

In the definition of $B_{v, \tau}[i,j,k]$ we assumed complete independence among 
the influence diffusion processes in the different  subtrees of $F(v,i),$ so it holds that 
$$B_{v,\tau}^{\textsc{i}}[i,j,k] = \max_{0 \leq a \leq j} \Big\{B_{v, \tau}[i-1, a, k] + B_{v, \tau}[\{i\}, j-a, 0]\Big\}$$
because the absence of a constraint on whether or not $v_i$ is influenced is the same as putting no constraint on the influence
of $v_i$ towards $v$.

Hence, by Lemma \ref{lemma:claim1} we have 
\begin{eqnarray}
B_{v,\tau}^{\textsc{i}}[i,j,k] &=& \max_{\substack{0 \leq a \leq j \\ \phantom{a}}} \left \{ B_{v,\tau}[i{-}1,a,k] +
\max \left\{
\max_{\substack{\tau_i \in \{0,1,\dots, \lambda, \infty\} \\ \phantom{\MIS[v_1, j, \tau_i, t'(v_1)] } \\ \phantom{MIS[v_1, j, \tau, t'(v_1)]}}}  \MIS[v_i, j-a, \tau_i, t(v_i)], ~~~~~~~~~~~~~~~~~~~~~ \right. \right.  \nonumber \\
& & \left. \left. ~~~~~~~~~~~~~~~~~~~~~~~~~~~~~~~~~~~~
\max_{\substack{\tau_i \in\{\tau+1, \dots, \lambda\} \\ \MIS[v_i, j-a, \tau_i, t'(v_i)] >\\ \MIS[v_i, j-a, \tau, t'(v_i)]}}
 \MIS[v_i, j-a, \tau_i, t'(v_i)] \right\} \label{eqBmax-a}
 \right\}.
\end{eqnarray}

Analogously, because of  the complete independence among 
the influence diffusion processes in the different subtrees of $F(v,i),$ assumed in the definition of $B_{v, \tau}[i,j,k],$
it holds that 
$$B_{v,\tau}^{\textsc{ii}}[i,j,k] = \max_{0 \leq a \leq j} \Big\{B_{v,\tau}[i{-}1,a,\max\{k-w(v_i, v),0\}] + 
B_{v,\tau}[\{i\},j-a,w(v_i, v)] \Big\}$$
since constraining $v_i$ to be influenced before time $\tau$ is the same as requiring that its influence towards $v$ is at least $w(v_i, v)$ before time $\tau.$
Hence, using Lemma \ref{lemma:claim2} we have

\begin{equation} \label{eqBmax-b}
B_{v,\tau}^{\textsc{ii}}[i,j,k] =
\max_{0 \leq a \leq j} \left \{ B_{v,\tau}[i{-}1,a,\max\{k-w(v_i, v),0\}] +
\max_{\tau_i \in \{0,1,\dots, \tau-1\}}  \MIS[v_i, j-a, \tau_i, t(v_i)] \right\}\,.
\end{equation}

From (\ref{eq:B-1-0})-(\ref{eqBmax-b}), 
it follows that the computation of $B_{v,\tau}[\cdot,\cdot,\cdot]$ comprises $O(d  b t)$ values and 
each one is computed recursively in time $O(\lambda b)$. Hence we are able to compute it
in time $O(d\lambda b^2 t)$.
\end{proof}

\medskip

\noindent We now consider  the computation of $\MIS[v,b,\spec,t]$.  

\begin{lemma}\label{lemma3}
For each vertex $v$, each $b=0,1,\ldots,\beta$, and each $t\in\{t(v),t'(v)\}$, it is possible to compute $\MIS[v,b,\spec,t]$  in time $O(d \lambda b^2),$ where $d$ is the number of children of $v$.
\end{lemma}
\begin{proof}
By Proposition \ref{prop:Cmax} it is enough to show that we  can compute $C_v[d,b]$ in the given time bound. We will do this
by  recursively computing the values  $C_v[i,j]$ for each $i=1,2,\ldots,d$ and for each $j=0,1,\ldots,b,$  as follows.
\\
For $i=1$,
we have that for any budget $j,$ it holds that
\begin{eqnarray}
C_v[1,j] &=&  \max_{\substack{S \subseteq T(v_1) \\ c(S) \leq j}} |\A{S, \lambda} \cap T(v_1)| \\
&=& \max_{\tau_1 \in \{0, 1, \dots, \lambda, \infty\}} \max_{\substack{S \subseteq T(v_1) \\ c(S) \leq j \\ v_1 \in \A{S, \tau_1}}} |\A{S, \lambda} \cap T(v_1)| \\
&=& \max_{\tau_1 \in \{0, 1, \dots, \lambda, \infty\}} \MIS[v_1, j, \tau_1, t(v_1)]
\end{eqnarray}
where the first equality holds because in this case $v$, whose contribution to the state of $v_1$ should be ignored, can only be influenced
by $v_1$ itself, hence in order to get $C_v[1,j]$ it is enough to consider only the vertices influenced in $T(v_1)$. The remaining equalities are obtained by
standard algebraic manipulation.

\medskip

\noindent
Now let $i>1$. For the sake of conciseness, we will abuse our definition and use weight $0$ to indicate that the
influence of $v$ on its children is to be neglected. Then we can write

\begin{eqnarray}
C_v[i,j] &=& \max_{\substack{S \subseteq F(v,i) \\ c(S) \leq j \\ w(v, v_k) = 0, k=1,\dots, i}} |\A{S, \lambda} \cap F(v,i)| \\
&=& \max_{\substack{S \subseteq F(v,i) \\ c(S) \leq j \\ w(v, v_k) = 0, k=1,\dots, i}} \left\{ |\A{S \cap F(v, i-1), \lambda} \cap F(v,i-1)| \right. \nonumber\\
& & ~~~~~~~~~~~~~~~~~~~~~~~~~~~~~~~~~~+ \left.
         |\A{S \cap T(v_i), \lambda} \cap T(v_i)|  \right\} \\
&=& \max_{0 \leq a \leq j} \left\{ \max_{\substack{S_1 \subseteq F(v,i-1) \\ c(S_1) \leq a \\ w(v, v_k) = 0, k=1,\dots, i-1}}
    |\A{S_1 \cap F(v, i-1), \lambda} \cap F(v,i-1)| \right.  \nonumber\\
& & ~~~~~~~~~~~~\left.  +
         \max_{\substack{S_2 \subseteq T(v_i) \\ c(S_2) \leq j-a \\ w(v, v_i) = 0}}   |\A{S_2 \cap T(v_i), \lambda} \cap T(v_i)| \right\} \\
 &=& \max_{0 \leq a \leq j} \left\{C_v[i-1, a] + \max_{\tau_i \in \{0, \dots, \lambda, \infty\}} \MIS[v_i, j-a, \tau_i, t_i] \right\}.
\end{eqnarray}

The last equality follows by the definition of $C_v[i-1,a]$ and since, in perfect analogy with the proof of the case $i=1,$ we can show that
$$\max_{\substack{S_2 \subseteq T(v_i) \\ c(S_2) \leq j-a \\ w(v, v_i) = 0}}   |\A{S_2 \cap T(v_i), \lambda} \cap T(v_i)| =
\max_{\tau_i \in \{0, \dots, \lambda, \infty\}} \MIS[v_i, j-a, \tau_i, t_i].$$

\medskip

There are $O(db)$ values of $C_v[\cdot ,\cdot]$ and each one is computed recursively in time $O(\lambda b)$. Hence,
by (\ref{eq-cmaxmis}),   we are able to compute $\MIS[v,b,\spec,t]$ in time $O(d\lambda b^2)$.
\end{proof}

\medskip



Thanks to the four lemmata \ref{lemma:MIS-ABC}, \ref{lemma1}, \ref{lemma2}, and \ref{lemma3} above, and recalling that for each node $v \in V,$ $t(v) \leq W(v)+1$,  we have that for each node $v \in V,$  for each $b=0,1,\ldots,\beta$, for each $\tau=0,1,\ldots, \lambda, \spec$, and for $t \in \{t'(v),t(v)\}$, $\MIS[v,b,\tau,t]$ can be computed recursively in time $O(\lambda\beta^2d(v)W(v))$. Hence, the value
\begin{equation} \label{eq:MIS}
\max_{\tau \in \{0,1,\ldots,\lambda,\spec\}} \ \MIS[r,\beta,\tau,t(r)]
\end{equation}
 can be computed in time
 $$\sum_{v \in V} O(\lambda\beta^2d(v)W(v)) \times O(\lambda \beta)=O(\lambda^2\beta^3)\times\sum_{v \in V} O(d(v)W(v))=O(\Delta \lambda^2W\beta^3),$$
where $\Delta$ is the maximum node in-degree and $W=\max_{u \in V} \{W(u)\}$ is the sum of all edge weights.
Standard backtracking techniques can be used to compute the (optimal) target set of cost at most $\beta$ that
influences this maximum number of nodes in the same $O(\Delta \lambda^2W\beta^3),$ time.
This proves Theorem \ref{theorem-tree}.

\medskip


\medskip

In case the tree is unweighted, one can obtain more precise bounds on the
complexity of the algorithm. Indeed, 
 reasoning analogous to that performed before can be used to show that, on unweighted trees,
for each node $v \in V,$  for each $b=0,1,\ldots,\beta$, for each $\tau=0,1,\ldots, \lambda, \spec$, and for $t \in \{t(v)-1,t(v)\}$, the values $\MIS[v,b,\tau,t]$ can be computed recursively in time $O(\lambda\beta^2d(v)t(v))$. 
Also, on unweighted graphs, for each node $v \in V $ it holds that  $t(v) \leq d(v)+1$,
so the value in (\ref{eq:MIS})
 can be computed in time
 $$\sum_{v \in V} O(\lambda\beta^2d(v)^2) {\times} O(\lambda \beta)=O(\lambda^2\beta^3)\times\sum_{v \in V} O(d(v)^2)=O(\min\{n\Delta^2\lambda^2\beta^3,n^2\lambda^2\beta^3\}).$$
Hence we have the following Corollary to Theorem \ref{theorem-tree}.

\begin{corol}\label{cor-tree}
The {\sc $(\lambda, \beta)$-MIS} problem can be 
solved in  time $O(\min\{n\Delta^2\lambda^2\beta^3,$ $n^2\lambda^2\beta^3\})$ on  an unweighted tree with $n$ nodes and maximum degree $\Delta$.
\end{corol}

\newcommand\lA[1]{p_A(#1)}

\newcommand\lS[1]{p_S(#1)}

\section{$(\lambda,\beta)$-MIS on Weighted Paths and Cycles}\label{path-cycles}

The results of Section \ref{sec-trees} obviously include paths.
However, for paths, we are able to significantly strengthen the result following from Theorem~\ref{theorem-tree}
by developing a polynomial time solution for the
{$(\lambda, \beta)$-\textsc{MIS}} problem on this class of graphs.
Let $P_n=(V,E)$  be a path on $n$ nodes $v_1,v_2,\ldots, v_{n}$,
and  edges  $(v_i,v_{i+1})$ and $(v_{i+1},v_i)$, for  $i= 1,\ldots, n-1$.

\begin{theorem}\label{teo-path}
The
{$(\lambda, \beta)$-\textsc{MIS}} problem can be solved in time $O(n^2\lambda)$
on a weighted path $P_n$.
\end{theorem}

\begin{proof}{}
For $i\in \{1,\ldots, n-1\}$,
let us denote $t'(v_i) = \max\{t(v_i)-w(v_{i+1},v_i),0\}$, and let $t'(v_n) = t(v_n)$.
Let $ V(P_i)$ be the set of vertices of a path $P_i$.
For $i\in \{1,\ldots, n\}$,
 $j\in \{0,1,\ldots, n\}$,
$\tau\in\{0,1,\ldots, \lambda, \infty\}$, and
$t\in \{t(v_i), t'(v_i)\}$,
let
$f(i,j,\tau, t)$ denote the  minimum cost of a subset $S\subseteq V(P_i)$
such that if the influence diffusion process is run on $P_i$ with target set $S$, where the threshold of
each node $v_k$ with $k <i$ is $t(v_k)$, while the threshold of $v_i$ is set to  $t$,
then vertex $v_i$ is influenced within time $\tau$ and
at least $j$ vertices are influenced within time $\lambda$.
If such a set does not exist, we set $f(i,j,\tau,t ) = \infty$.
Furthermore, let $S(i,j,\tau, t)$ denote any set $S\subseteq V(P_i)$ attaining the value of
$f(i,j,\tau, t)$ (whenever this value is finite).

Notice that $f(n,k,\infty, t(v_n))$ equals the minimum cost of a subset $S\subseteq V(P_n)$
when the influence diffusion process is run on the input path with target set $S$ such that at least $k$ nodes are influenced
within $\lambda$ steps. Therefore, to solve the {$(\lambda, \beta)$-\textsc{MIS}} problem on $P_n$,
it suffices to find the maximum value of $k\in \{0,1,\ldots, n\}$ such that
$f(n,k,\infty, t(v_n))\le \beta$. An optimal solution will then be given by $S(n,k,\infty, t(v_n))$.

We now explain how all of the values of $f(i,j,\tau, t)$ and the corresponding sets
$S(i,j,\tau, t)$ can be computed in time $O(n^2\lambda)$.

First, observe that $f(i,j,\tau, t) = \infty$ if and only if $j>i$.
Indeed, if $j>i$ then the condition that at least $j$ elements out of $i$ are influenced within time $\lambda$ clearly cannot be fulfilled.
On the other hand, if $j\le i$, then $S = V(P_i)$ is a feasible solution for the problem defining $f(i,j,\tau, t)$.
Hence, in what follows, we will assume that $j\le i$ for every $4$-tuple $(i,j,\tau,t)$ under consideration.

We proceed in order of increasing values of $i$ and prove a sequence of claims.

\begin{claim}\label{eq:path-dynamics-i=1}
For $i = 1$, we have $$S(1,j,\tau, t) = \left\{
                                          \begin{array}{ll}
                                            \emptyset & \hbox{if $(j = 0 \textrm{ AND } \tau = \infty)  \textrm{ OR } (\tau \ge 1\textrm{ AND } t =

0)$} \\
                                            \{v_1\} & \hbox{otherwise,}
                                          \end{array}
                                        \right.$$
and
$f(1,j,\tau, t) = c(S(i,j,\tau, t)).$
\end{claim}

\noindent{\it Proof.}
For $j = 0$  and $\tau = \infty$, both constraints, the one specifying that $v_1$ should be influenced within time $\tau$, and the one
specifying that at least $j$ vertices become influenced within time $\lambda$, are vacuous. Therefore
$S=\emptyset$ is an optimal solution in this case.
If $\tau\ge 1$ and $t = 0$, then $v_1$ will
become influenced at time $1$ (which is not more than $\tau$), which also implies that
the constraint $|\Active[S,\lambda]\cap\{v_1\}|\ge j$ will
be satisfied for any $j\in \{0,1\}$ independently of $S$, which implies that $S=\emptyset$ is optimal. Suppose now that
($\tau \le \lambda$ or $j =
1$) and ($\tau = 0$ or $t >0$).
It suffices to show that the empty set is not a feasible solution.
Suppose by way of contradiction that it is. Then $\tau\ge 1$ and consequently $t>0$,
which implies that vertex $v_1$ will not become influenced. Consequently, neither
$\tau \le \lambda$ nor $j = 1$ are possible, a contradiction.\hfill$\blacktriangle$

\bigskip
Now let $i>1$, and suppose inductively that $f(i',j,\tau, t)$
and the corresponding target sets were already computed for all
$i'<i$ and all suitable values of $j$, $\tau$, and $t$.
In the next sequence of claims, we will show how to compute $S(i,j,\tau,t)$ and
$f(i,j,\tau, t)$ (for all suitable values of $j$, $\tau$, and $t$).
First we deal with the cases when $\tau = 0$.

\begin{claim}\label{eq:path-dynamics-0}
If $i>1$ and $\tau = 0$, then
$S(i,j,0,t) = S(i-1,\max\{j-1,0\}, \infty, t'(v_{i-1}))\cup\{v_i\}$
and
$f(i,j,0,t) = c(S(i,j,0,t))$.
\end{claim}

\noindent{\it Proof.}
The fact that $\tau = 0$ implies that $v_i$ must be taken in the corresponding target set,
that is, $v_i\in S(i,j,0, t)$.
It suffices to prove that $$f(i,j,0,t) = f(i-1,\max\{j-1,0\}, \infty, t'(v_{i-1}))+c(v_i)\,.$$

Let $S' = S(i-1,\max\{j-1,0\}, \infty, t'(v_{i-1}))$.
To show the  inequality ``$\le$'',
it suffices to argue that
when running the influence diffusion process in $P_i$ with target set $S = S'\cup \{v_i\}$, we have
$|\Active[S,\lambda]\cap V(P_i)|\ge j$. 
Indeed, assuming this property, we have that
$$f(i,j,0,t) \le c(S) = c(S')+c(v_i) = f(i-1,\max\{j-1,0\}, \infty, t'(v_{i-1}))+c(v_i)\,,$$
where the first inequality holds by definition of $f(i,j,0,t)$, the first equality holds by the definition of $S$,
and the last equality holds by the definition of $S'$. 
To justify the above claim, note that when running the influence diffusion process  in $P_{i-1}$ with target set $S'$,
at least $j-1$ vertices get influenced within $\lambda$ rounds. These vertices will also get influenced
within $\lambda$ rounds by the influence diffusion process in $P_i$ with target set $S$; in addition, $v_i$ will be
influenced since it belongs to the target set.

Similarly, to show the reverse inequality, ``$\ge$'',
it suffices to argue that when running the influence diffusion process  in $P_{i-1}$ with target set $S' = S(i,j,0,t)\setminus\{v_i\}$,
and with the threshold of $v_{i-1}$ set to $t'(v_{i-1})$,
at least $j-1$ vertices get influenced within $\lambda$ rounds. This follows from the observation that
for every $k$ with $1\le k\le i-1$, vertex $v_k$
gets influenced within $\lambda$ rounds in $P_i$ by the target set $S(i,j,0,t)$ if and only if
it gets influenced within $\lambda$ rounds in $P_{i-1}$ by the target set $S'$ with the modified threshold of $v_{i-1}$.
\hfill$\blacktriangle$

\medskip

Now, we handle the case when $\tau\neq 0$ and $t = 0$.

\begin{claim}\label{eq:path-dynamics-tau-t=0}
If $i>1$, $\tau\neq 0$, and $t = 0$,
then
$$S(i,j,\tau,0) = S(i-1,\max\{j-1,0\}, \infty, t')\,,$$
where $$t'=\left\{
             \begin{array}{ll}
               t'(v_{i-1}) & \hbox{if $\lambda>1$} \\
               t(v_{i-1}) & \hbox{otherwise,}
             \end{array}
           \right.$$
and $f(i,j,0,t) = c(S(i,j,0,t))$.
\end{claim}

\noindent{\it Proof.}
Since $t = 0$, vertex $v_i$ will become influenced at time $1$, no matter what the target set is.
If in addition $\lambda > 1$, then vertex $v_i$ can help to influence $v_{i-1}$ at times between $2$ and $\lambda$.
It suffices to prove that
$f(i,j,\tau,0)= f(i-1,\max\{j-1,0\}, \infty, t')\,.$
To show that $$f(i,j,\tau,0)\le f(i-1,\max\{j-1,0\}, \infty, t')\,,$$
note that in $P_{i-1}$, the influence diffusion process
with the target set $S(i-1,\max\{j-1,0\}, \infty, t')$
influences at least $j-1$ vertices
within $\lambda$ rounds.
These vertices, together with $v_i$, form a set of at least $j$ vertices
influenced within $\lambda$ rounds in $P_i$ by the same target set.
Conversely,
to show that $$f(i-1,\max\{j-1,0\}, \infty, t')\le f(i,j,\tau,0)\,,$$
observe that
the influence diffusion process in $P_i$ with target set
$S(i,j,\tau,0)$ influences at least $j-1$ vertices within $P_{i-1}$ within $\lambda$ rounds.
Moreover, if vertex $v_{i-1}$ is not in the target set but gets influenced within $\lambda$ rounds, then
this vertex will also get influenced when the influence diffusion process is run in $P_{i-1}$ with target set
$S(i,j,\tau,0)$ (which does not contain $v_i$, by optimality and the fact that costs are positive) and the threshold of $v_{i-1}$ set to $t'$.
This establishes the second inequality and proves the claim.\hfill$\blacktriangle$

\medskip
The remaining case is when $t > 0$, which is split into two further subcases, depending on whether $\tau$ is finite on not.

\begin{claim}\label{eq:path-dynamics-tau-t>0}
If $i>1$, $\tau\in \{1,\ldots, \lambda\}$, and $t > 0$,
then
$$f(i,j,\tau,t) = \left\{
                    \begin{array}{ll}
                      \min\{f(i,j,0, t), f(i-1,\max\{j-1,0\}, \tau-1, t(v_{i-1}))\} & \hbox{if $w(v_{i-1},v_i)\ge t$} \\
                      f(i,j,0, t) & \hbox{otherwise,}
                    \end{array}
                  \right.
\,$$
and
the set $S(i,j,\tau,t)$ is defined in the obvious way depending on where the minimum is attained.
\end{claim}

\noindent{\it Proof.}
Since $t > 0$, there are exactly two ways in which vertex $v_i$ can become influenced within time~$\tau$: either $v_i$ is
placed  in the target set, or it becomes influenced because $w(v_{i-1},v_i)\ge t$ and its unique neighbor, vertex $v_{i-1}$, becomes
influenced within time $\tau-1$.
This observation, together with arguments similar to those used in the proofs of previous claims, establishes the claim.~\hfill$\blacktriangle$

\medskip

Finally, for $\tau= \infty$ and $t>0$ we have the following.

\begin{claim}\label{eq:path-dynamics-tau-infty}
If $i>1$, $\tau= \infty$, and $t > 0$,
then
$$f(i,j,\infty, t) = \min\left\{\min_{0\le \tau'\le \lambda}f(i,j,\tau', t), f(i-1,j,\infty, t(v_{i-1}))\right\}\,,$$
and the set $S(i,j,\infty, t)$ is computed in the obvious way depending on where the minimum in the above expression is attained.
\end{claim}

\begin{sloppypar}
\noindent{\it Proof.}
Note that by definition of
$f(i,j,\tau, t)$, we have
$f(i,j,\infty, t) \le \min_{0\le \tau'\le \lambda}f(i,j,\tau', t)$.
Also, if $j\le i-1$, then running the influence diffusion process in $P_i$ with target set $S = S(i-1,j,\infty, t(v_{i-1}))$
results in at least $j$ influenced vertices (already within $V(P_{i-1})$),
showing that
$f(i,j,\infty, t) \le f(i-1,j,\infty, t(v_{i-1}))$.
This establishes that
$f(i,j,\infty, t) \le \min\left\{\min_{0\le \tau'\le \lambda}f(i,j,\tau', t), f(i-1,j,\infty, t(v_{i-1}))\right\}\,.$
\end{sloppypar}

For the converse direction, take an optimal solution $S = S(i,j,\infty, t)$, and consider the influence diffusion process in $P_i$ with target set
$S$
for $\lambda$ rounds.
Let $\tau_i$ be the time at which $v_i$ is influenced (with $\tau_i = \infty$ if $v_i$ is not influenced within $\lambda$ rounds).
If $\tau_i$ is finite, then
$f(i,j,\infty, t)= f(i,j,\tau_i, t)$, and
hence $\min_{0\le \tau'\le \lambda}f(i,j,\tau', t)\le f(i,j,\tau_i, t)= f(i,j,\infty, t)$.
If $\tau_i = \infty$, then $v_i$ is not influenced within time $\lambda$, which implies that
$S\subseteq V(P_{i-1})$, $j\le i-1$, and running the influence diffusion process in $P_{i-1}$ with target set $S$
for $\lambda$ rounds results in at least $j$ influenced vertices, showing that in this case
$f(i-1,j,\infty, t(v_{i-1}))\le f(i,j,\infty, t)$.
This proves the claim.
\hfill$\blacktriangle$

\medskip
To justify the time complexity of the resulting algorithm,  note that
there are $O(n^2\lambda)$ $4$-tuples $(i,j,\tau, t)$. Using the above formulas,
the corresponding optimal values of $f(i,j,\tau, t)$ and target sets $S(i,j,\tau, t)$ (in case of feasible problems)
can be computed in time $O(n^2\lambda)$.
\end{proof}


We conclude this section by extending our result for paths to cycles.
We denote by $C_n$ the cycle on $n\ge 3$ nodes that consists  of the path $P_n$ augmented with the edges $(v_1,v_n)$ and $(v_n,v_1)$.

\begin{theorem}\label{teo-cycle}
The {$(\lambda, \beta)$-\textsc{MIS}} problem can be solved in time $O(n^3\lambda)$ on a weighted cycle $C_n$.
\end{theorem}

\begin{proof}
We describe how to reduce the problem to solving at most $n$ instances of the
{$(\lambda, \beta)$-\textsc{MIS}} problem on paths. The result will then follow from
Theorem~\ref{teo-path}.

We compute the set $I$ of all indices $i\in \{1,\ldots, n\}$ such that $c(v_i)\le \beta$.
We set $S_0 = \emptyset$, and
compute, for each $i\in I$, a target set $S_i$ with $v_i\in S_i$ such that
the number of nodes influenced within $\lambda$ rounds when running the influence diffusion process on $C_n$ with $S$,
over all sets $S$ containing $v_i$ and of total cost at most $\beta$,
is maximized for $S_i$. Once the sets $S_i$ for $i\in I$ are computed,
computing the number of influenced nodes within $\lambda$ rounds for each target set $S_i$, where $i\in I\cup\{0\}$,
can be used to determine an optimal solution.

For each $i\in I$, the problem of computing $S_i$ can be reduced to an instance of the
{$(\lambda, \beta)$-\textsc{MIS}} problem on the $(n-1)$-vertex path $C_n-\{v_i\}$, as follows.
Since we assume that $v_i\in S_i$, we reset the threshold of $v_{j}$ for $j\in \{i-1,i+1\}$ (indices modulo $n$)
to $t'(v_j) = \max\{t(v_j)-w(v_{i},v_j),0\}$.
We delete vertex $v_i$ from the graph (thus obtaining a path), reduce the budget to $\beta-c(v_i)$, and keep the latency bound $\lambda$ unchanged.
This way, it can be readily seen that we obtain a weighted path instance of the {\sc $(\lambda, \beta)$-MIS}
problem such that if $S$ is an optimal solution for this instance, then
$S_i = S\cup \{v_i\}$ has the desired property.

Together with Theorem~\ref{teo-path}, we obtain the claimed result.
\end{proof}

\section{Concluding Remarks and Open Problems}
We  considered the problems of selecting a {bounded} cost
subset of nodes in (classes of)  networks such that the influence they
spread in a {fixed}  number of rounds is the {highest} among
all subsets of the same bounded cost.
It is not difficult to see that our techniques can also solve closely related 
problems in  the same classes of graphs considered in this paper.
 For instance, one could fix a requirement $\alpha$ and ask for
the \emph{minimum} cost  target set such that after $\lambda$ rounds the number
of influenced nodes in the network is at least $\alpha$. Or, one could
fix a budget $\beta$ and a requirement $\alpha$, and ask about the \emph{minimum}
number $\lambda$ such that there exists a  target set of cost  at most $\beta$
that influences at least $\alpha$ nodes in the network within $\lambda$ rounds
(such a minimum $\lambda$ could also  be equal to  $\infty$, meaning that a target set 
with the desired properties does not exist).

{To the best of our knowledge, there are no results for the problems we considered in this paper
for ``less structured'' network models, like  small world graphs or  exponential random graphs and, in general,
 for models that better  capture  real-world properties of social networks. 
We plan to investigate these problems    in future work.}

Another interesting extension of our results would be  to consider the case in which
there is a numerical value $p(\cdot)$ associated with each node $v$ in  the network, measuring  the \emph{profit}  that an
advertiser, say,  would gain from  convincing $v$   to adopt a product.
This numerical value could be related, for instance, to the purchasing power (or the purchasing inclination)
of the individual.
In this scenario, one would be interested  in  finding a target set $S$ of bounded   cost
such that the sum of the profits associated with
influenced nodes, 
 computed as 
$$\sum_{v\in \Active[S,\lambda]}p(v),$$
  is the {highest} among
all subsets of the same bounded costs. We leave this problem open for future investigations.

\section*{Acknowledgments}
{The authors would like to thank the anonymous referees for their careful reading of the manuscript and for their many  valuable comments.}

\end{document}